\documentclass[10pt,letterpaper,twocolumn]{article} 

\usepackage{ol2}
\usepackage[draft]{hyperref}
\usepackage{amsmath}
\usepackage{amsfonts,amsthm}

\newcommand{\mbf}[1]{\ensuremath{\mathbf{#1}}}
\newcommand{\mbb}[1]{\ensuremath{\mathbb{#1}}}

\newtheorem{theorem}{Theorem}
\newtheorem{definition}{Definition}
\usepackage{amssymb}
\usepackage{fontenc}
\usepackage{times}
\usepackage{mathptmx}
\begin{document}

\twocolumn[ 

\title{Odd-symmetry phase gratings produce optical nulls uniquely insensitive to wavelength and depth}


\author{Patrick R. Gill}

\address{Rambus Labs, 1050 Enterprise Way, Suite 700, Sunnyvale, California, 94089, USA}

\begin{abstract}I present the analysis of a new class of diffractive optical element, the odd-symmetry phase grating, which creates wavelength- and depth-robust features in its near-field diffraction pattern.   \end{abstract}


 ] 

\section{Introduction}
Near-field diffractive optical imagers, such as arrays of angle-sensitive pixels (ASPs) that exploit the Talbot effect using integrated CMOS amplitude diffraction gratings \cite{wang2009light,gill2011microscale}, show great promise in enabling the construction of unprecedentedly-small optical sensors.   However, technical and fundamental obstacles limit the high-resolution performance of standard ASP arrays. \cite{gill2012scaling,wang2012light}  Three fundamental obstacles to ASPs' application in high-resolution optical sensing under standard illumination are their sensitivity to manufacturing errors, their transfer function phase-reversals caused by changes in wavelength, and their decreasing area efficiency at higher resolutions.  

Here, I present an analysis of the near-field diffraction caused by odd-symmetry gratings,  a class of diffractive optical element fundamental to the operation of a new class of ultra-miniature imager. \cite{GillStork:13}  Odd-symmetry gratings, unlike ASPs, exhibit wavelength- and depth-robust, compact, angle-dependent null planes under lines of odd symmetry, defined as follows.

\begin{definition}
A null under a line of odd symmetry is ``robust'' if, for normally-incident light of any wavelength $\lambda$, the light intensity on planes beneath the odd-symmetry line is 0.
\end{definition}

Nulls produced by ASPs are not robust due to their acute sensitivity to wavelength and manufacturing depth.  

Like the null at the center of an optical vortex \cite{kivshar2001optical,ostrovsky2013generation}, robust nulls created by odd-symmetry phase gratings are sensitive neither to the depth below the phase element nor to the wavelength of light.  Unlike optical vortices which produce line-shaped nulls, odd-symmetry phase gratings produce plane-shaped nulls.

\section{Geometry and  scalar diffraction}

\begin{figure}[htb]
 \begin{center}
  \includegraphics[width=0.99\columnwidth]{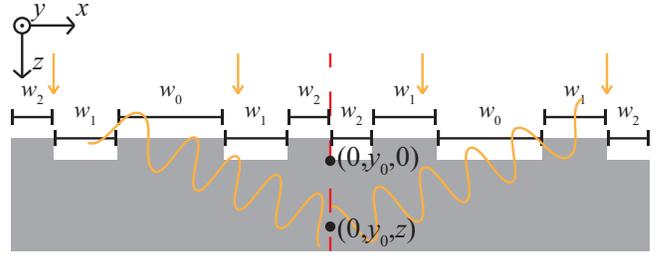}
  \end{center}
\caption{An odd-symmetry, binary phase grating shown in cross-section.  A phase grating at the intersection of two optical media introduces a phase delay of a half wavelength between light passing through a thick versus a thin portion of the grating.  One robust null is shown as a dashed vertical line; if the grating were to repeat then both the left and right borders would also exhibit robust nulls.  The optical media have refractive indices and dispersions such that the phase delay is roughly constant across $\lambda$s of interest.  The grating in this example has only two depths (making it easy to manufacture), whereas in general the phase delay at any point need not belong to a discrete set.}
\label{fig:Geo}
\end{figure}

The electric field amplitude at a specific point $(0,y_0,z>0)$ below the phase grating (see Fig. \ref{fig:Geo}) can be found by integrating the contributions to that point caused by light passing through all locations of the phase grating above.  I will use $y_0$ to denote the $y$ coordinate of the point where the electric field amplitude is observed, reserving $y$ for  positions at the phase grating in the same coordinate system.  
Let the time-varying optical electric field induced by light of polarization $\psi$ and wavelength (in the medium) $\lambda$ at a point $(0,y_0,z)$ be $\Re \left(E(\lambda,y_0,z,\psi) e^{\frac{-2\pi i c t}{n \lambda}}  \right)$, where $c$ is the speed of light and $n$ is the refraction index of the medium.  The light intensity at $(0,y_0,z)$ is proportional to $\left|E(\lambda,y_0,z,\psi)\right|^2 $.  The expression for the complex electric field amplitude $E(\lambda,y_0,z,\psi)$ for normally-incident light is given by scalar diffraction:
\begin{equation}
E_{\lambda,y_0,z,\psi}\!=\!\iint\limits_{P}C(x,y-y_0,z)e^{2\pi i \frac{r}{\lambda}}   e^{i \phi(x,y,\psi,\lambda)} dy \,dx. \label{fullEM}
\end{equation}
$P$ is the plane of the grating, $r = \sqrt{x^2 + (y-y_0)^2 + z^2}$, and $C(x,y-y_0,z)$ represents the magnitude of the Green's function governing the coupling between the surface of the grating at $(x,y,0)$ and the point $(0,y_0,z)$.  $C$ is strictly positive on $r < \infty$ and $C(x,y-y_0,z) =  C(-x,y-y_0,z)$.  The real part of $\phi(x,y,\psi,\lambda)$ is the phase delay (in radians) introduced by the grating at position $(x,y,0)$ for light of polarization $\psi$ and wavelength $\lambda$, while the imaginary part of $\phi$ represents attenuation, which we assume to be finite or 0.  Exploiting the reflection symmetry of most terms in Eq. \ref{fullEM} about $x=0$, it is possible to write an expression for $E$ integrating over the half-plane $H$ where $x>0$: 
\begin{equation}
E_{\lambda,y_0,z,\psi}\!=\!\iint\limits_{H}C(x,y-y_0,z)e^{2\pi i \frac{r}{\lambda}}   \underbrace{\left(e^{i \phi(x)} + e^{i \phi(-x)}\right)}_{\text{Grating effects}} dy \, dx \label{halfEM} 
\end{equation}
where $\phi$'s dependence on $(y,\psi,\lambda)$ has been omitted for compactness and the terms showing the effects of the gratings have been explicitly marked.  Let us define odd-symmetry phase gratings as follows.
\begin{definition}
A phase grating has odd symmetry along the $y$-axis if Eq. \ref{eq:oddSym} holds:
\begin{equation}
\phi(x,y,\lambda,\psi) \! = \! \phi(-x,y,\lambda,\psi) + \pi + 2\pi m, \, \, m\; \in \mbb{Z} \label{eq:oddSym}
\end{equation}
where $\mbb{Z}$ denotes the set of integers.  \end{definition} Phase profiles similar to these have been explored in Dammann gratings \cite{morrison1992symmetries}, however to my knowledge this is their first use in near-field diffraction elements.  Figure \ref{fig:Geo} illustrates a phase grating with only two heights ({\em i.e.} a binary grating), while gratings can take multiple, or even piecewise continuous, heights and still conform to Eq. \ref{eq:oddSym}.  

There are three free parameters in the binary grating of Fig. \ref{fig:Geo}: the lengths $w_0$, $w_1$ and $w_2$.  In general, one can construct a binary odd-symmetry grating with any number of such lengths, so long as thick and thin grating segments alternate as shown in Fig. \ref{fig:Geo}.  Repeating binary odd-symmetry gratings (such as the design shown in Fig. \ref{fig:Geo}) with $n+1$ free parameters create repeated odd-symmetry planes spaced apart by a distance $w_0 + \sum\limits_{j=1}\limits^{j=n}2w_j$.

Ignoring dispersion, phase delays will be proportional to the reciprocal of the wavelength.  However, by pairing a high-dispersion, low-$n$ optical plastic above a low-dispersion, high-$n$ optical glass, the phase delay introduced by the grating can be made to be approximately wavelength-independent, as seen in Fig. \ref{fig:OKn}.

\begin{figure}[htb]
 \begin{center}
  \includegraphics[width=0.99\columnwidth]{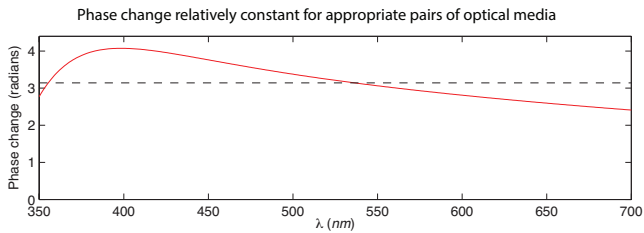}
  \end{center}
\caption{Phase delay (solid red) induced by a 0.9-micron-tall phase grating made from a high-dispersion, low-$n$ optical plastic above a low-dispersion, high-$n$ optical glass.  Phase delays are roughly equal to $\pi$ (dashed line) for all visible light.}
\label{fig:OKn}
\end{figure}

\section{Proof that odd symmetry is sufficient to produce robust nulls}
\label{sec:suf}
At all points on the half-plane $(0,y_0,z),\, z>0$, Eq. \ref{eq:oddSym} guarantees that a for every contribution to $E$ from the left ($x<0$), there is a contribution from the right ($x>0$) of equal magnitude but opposite phase.  As each pair of contributions cancel, they sum to 0.  To aid our proofs of the equivalence of odd symmetry and robust nulls, I will introduce two functions $q$ and $p$ that loosely describe the in- and out-of-phase contributions from pairs of points reflected about $x=0$.  
\begin{theorem}
If a phase grating satisfies Eq. \ref{eq:oddSym}, it has a robust null at $x = 0$.
\end{theorem}
\begin{proof}
To see why odd symmetry produces robust nulls, introduce the following functions:
\begin{equation}
\label{defp} p(x) \! \equiv \! \frac{\phi(x) + \phi(-x) + \pi}{2}
\end{equation}
\begin{equation}
\label{defq} q(x) \! \equiv \! \frac{\phi(x) - \phi(-x) - \pi}{2},
\end{equation}
where  dependence on $(y,\psi,\lambda)$ has been omitted for succinctness.  The following identities obtain:
\begin{equation}
\label{sumpq} \phi(x) \! = \! p(x)+q(x)
\end{equation}
\begin{equation}
\label{pneg} p(-x) \! = \! p(x)
\end{equation}
\begin{equation}
\label{qneg} q(-x) \! = \! -q(x) - \pi 
\end{equation}
\begin{equation}
\label{allTogether} \left(e^{i \phi(x)} + e^{i \phi(-x)}\right) \! = \! 2ie^{i p(x)}\sin(q(x)).
\end{equation}
If and only if Eq. \ref{eq:oddSym} is satisfied, $q(x)$ evaluates to $m\pi,\, \, m\; \in \mbb{Z}$ and $\sin(q(x)) = 0$, so by Eq. \ref{allTogether}, $\left(e^{i \phi(x)} + e^{i \phi(-x)}\right)=0$.  Substituting into Eq. \ref{halfEM}, we see that if Eq. \ref{eq:oddSym} is satisfied, $E_{\lambda,y_0,z,\psi} = 0$ regardless of $z$, $y_0$, $\lambda$ or $\psi$, and the null is robust. \qedhere
\end{proof}
Therefore, adherence to Eq. \ref{eq:oddSym} is sufficient to create robust nulls in the near-field diffraction patterns produced by phase gratings.  Diffraction-based optical elements made with odd-symmetry gratings therefore do not exhibit phase-reversing near-field nulls with changing wavelength or depth.  

If the phase grating has multiple nulls, the spacing between any adjacent pair being at least several microns, then it is straightforward to construct multiple photodiodes per period of the overlying phase grating.  I discuss elsewhere \cite{GillStork:13} how this arrangement can lead to an angle-sensitive photosensor with much better information density than ASP-based optical sensors.

\section{Proof that odd symmetry is necessary to produce robust nulls}
In Sect. \ref{sec:suf}, we saw that odd symmetry is sufficient to induce robust nulls.  Here, I prove that it is also necessary, a useful result in that it can dramatically limit the design  space needed to be considered for designing wavelength- and depth-robust phase gratings.  
\begin{theorem}
If a phase grating has a robust null at $x = 0$ extending for an open set of $y_0$ and its attenuation is less than infinite everywhere, its phase retardation and attenuation $\phi$ satisfies Eq. \ref{eq:oddSym} on every open set of grating locations.
\end{theorem}
The constraint that the grating's attenuation is never infinite is a consequence of the fact that the phase of a contribution of $0$ magnitude is not relevant to the observed electric field magnitude.  A straightforward extension of the proof below covers the limit where certain regions of the grating (also symmetric about $x=0$) completely block light.  The ``open set'' restriction is also useful to exclude sets of measure 0 and points infinitely far from $(0,y_0,z)$.
\begin{proof}
I will construct a proof by contradiction.  Assume the grating has a robust null, yet violates Eq. \ref{eq:oddSym} on some open set of grating locations.  I will factor the terms in the integrand of Eq. \ref{halfEM} into one portion $F$ that is 0 only if Eq. \ref{eq:oddSym} is satisfied, and a second portion $G$ which is always nonzero in magnitude.  I will then use Parseval's theorem to show that assuming a null is robust actually implies $F$ must be 0 except on sets of measure 0, violating the assumption that Eq. \ref{eq:oddSym} is not satisfied on some open set.  
Define $F$ and $G$ as follows: $F \equiv \sin(q(x,y,\lambda,\psi))$ and $G \equiv 2iC(x,y-y_0,z)e^{ip(x,y,\lambda,\psi)}$.  The Green's function $C(x,y-y_0,z)$ has a finite magnitude except at $\infty$, and by assumption that the grating's attenuation is less than infinite, $ \Im \left(p(x,y,\lambda,\psi)\right) < \infty$, making the magnitude of  $e^{ip(x,y,\lambda,\psi)}$ also finite.  Therefore, $|G|$ is nonzero except at $\infty$.  By Eqs. \ref{halfEM} and \ref{allTogether},
\begin{equation}
E_{\lambda,y_0,z,\psi}\!=\!\iint\limits_{H}F(x,y) e^{2\pi i  \frac{r}{\lambda}} G(x,y) \, dy \, dx \label{FG1} 
\end{equation}
Introduce the change of variables $\mbf{\nu}  \equiv \hat{r} \frac{c}{\lambda}$ where $\hat{r}$ is the unit vector in the direction of $(x,y)$ and $|\mbf{\nu}|$ is the frequency of the light of wavelength $\lambda$.  Also introduce a change of variables $\mbf{r'} \equiv \hat{r} \left( \frac{\sqrt{x^2 + (y-y_0)^2 + z^2}}{c} - \frac{z}{c}\right)$, making new functions $F'$ and $G'$, which are functions of $\mbf{r'}$ and $\mbf{\nu}$.  Note that \mbf{r'} and \mbf{\nu} are vector quantities, shown in boldface.  Let $F'(\mbf{r'})$ take the exact value of $F(x,y)$ at the $\mbf{r'}$ corresponding to $(x,y)$, but scale $G'$ such that $G'(\mbf{r'}) d\mbf{r'} = 2\pi\, G(x,y)\, dx\, dy$.  Other than at $\mbf{r'} = \mbf{0}$ (a set of measure 0), $F'$ and $G'$ are finite precisely where $F$ and $G$ are, open sets on the $x$-$y$ plane are mapped to open sets on $\mbf{r'}$, and $|G'|$ is finite except at $\mbf{\infty}$.  Equation \ref{FG1} becomes:
\begin{equation}
E_{\nu,y_0,z,\psi}\!=\!\frac{1}{2\pi}\iint\limits_{H}F'(\mbf{r'}) e^{2\pi i \mbf{\nu} \cdot \mbf{r'}} G'(\mbf{r'}) d\mbf{r'} \label{FG2}. 
\end{equation}
Extending $F'$ and $G'$ out of the half-plane $H$ as follows:
 \begin{eqnarray}
 F''(\mbf{r'}) & \equiv &
 \left\{ 
 	\begin{array}{ll}
		F'(\mbf{r'}) & \mbox{if } x > 0  \\
		0 & \mbox{if } x \le 0
	\end{array}	
\right. \nonumber \\
 G''(\mbf{r'}) & \equiv &
 \left\{ 
 	\begin{array}{ll}
		G'(\mbf{r'}) & \mbox{if } x > 0  \\
		1 & \mbox{if } x \le 0
	\end{array}	
\right. \nonumber
 \end{eqnarray}
permits Eq. \ref{FG2} to be extended over the full plane $P$, leading to a Fourier transform in $\nu$:
\begin{equation}
E_{\nu,y_0,z,\psi}\!=\!\frac{1}{2\pi}\iint\limits_{P}F''(\mbf{r'}) e^{2\pi i \mbf{\nu} \cdot \mbf{r'}} G''(\mbf{r'}) d\mbf{r'} \label{FG3}. 
\end{equation}
Note that $|G''|$ is finite except at $\infty$ and \mbf{0}. By Parseval's theorem, Eq. \ref{FG3} implies that
\begin{equation}
\iint\limits_{P}\left|E_{\nu,y_0,z,\psi}\right|^2 d\nu\!=\!\iint\limits_{P}\left| F''(\mbf{r'}) G''(\mbf{r'})\right|^2 d\mbf{r'} \label{FG4}. 
\end{equation}
The integral of $E_{\nu,y_0,z,\psi}$ on the circle $|\nu| = \frac{c}{\lambda}$ gives the amplitude of the electric field due to light at wavelength $\lambda$ observed at $(0,y_0,z)$, which is $0$ by assumption.  
Since $G''$ is finite on all open sets of $P$ except at the origin, either $F''$ must be 0 on all open sets (and the grating has odd symmetry) or $F''$ is nonzero and $E_{\nu,y_0,z,\psi}$ is nonzero but its integral over all circles $|\nu| = \frac{c}{\lambda}$ is zero (implying rotational symmetry in $\phi$ about $(x,y_0)$).  This second possibility cannot be true by assumption that the null persists over an open set of $y_0$, and contiguous points cannot all have nontrivial rotational symmetry.  Therefore $F''$ is 0 on all open sets and the grating must satisfy Eq. \ref{eq:oddSym} on all open sets, violating the assumptions and achieving proof by contradiction. \qedhere
\end{proof}

Therefore, the only way to achieve robust null halfplanes in a near-field diffraction pattern with a phase grating is if that phase grating has one or more odd-symmetry lines.  
\section{Applications and conclusions}
This paper has shown that odd-symmetry phase gratings produce wavelength- and depth-robust null halfplanes under normally-incident light, and conversely that the only way of manufacturing wavelength-robust null halfplanes is to construct an odd-symmetry grating.  

Odd-symmetry gratings promise many new applications.  They are as small as other diffraction-scale elements such as ASPs, yet they overcome many of the latter's limitations.  Thus, we expect odd-symmetry gratings might enable new classes of optical sensors and light sources, and we are actively investigating optical sensing, illumination and imaging applications using these gratings both alone and in combination with other optical elements.  

\section*{Acknowledgement}
I would like to thank David G. Stork for insightful comments and guidance in the formulation of this paper and its proofs.

\end{document}